\documentclass{llncs}
\usepackage{amssymb}
\usepackage{amsmath}

\usepackage{algorithm}
\usepackage{algpseudocode}
\usepackage{nicefrac}
\usepackage{url}
\usepackage{color}
\usepackage[american]{babel}

\usepackage{graphicx}
\usepackage{subfigure}

\newtheorem{observation}{Observation}
\newtheorem{invariant}{Invariant}

\newcommand{\calP}{{\ensuremath{\cal P}}}
\newcommand{\calX}{{\ensuremath{\cal X}}}
\newcommand{\calB}{{\ensuremath{\cal B}}}

\newcommand{\leaveout}[1]{}
\renewcommand{\paragraph}[1]{\medskip\noindent{\bf #1}}

\renewcommand{\medskip}{\smallskip}
\renewcommand{\int}{{\ensuremath{\rm int\,}}}

\usepackage{lineno}

\date{}
\title{
Triangulating planar graphs \\ while keeping the pathwidth small
\thanks{
Research was supported by NSERC and done while visiting
Universit\"at Salzburg.    
Many thanks to Jasine Babu for 
sharing her manuscript of what later became \cite{BBC+14}, and
the referees of an earlier version of this paper for helpful comments.
}
}
\author{Therese Biedl}
\institute{David R.~Cheriton School of Computer
Science, University of Waterloo, \\ Waterloo, Ontario N2L 1A2, Canada.
\email{biedl@uwaterloo.ca}}

\begin{document}

\maketitle
\begin{abstract}
Any simple planar graph can be triangulated,
i.e., we can add edges to it, without adding multi-edges, such
that the result is planar and all faces are triangles.
In this paper, we study the problem of triangulating a planar
graph without increasing the pathwidth by much.  

We show that if a planar graph has pathwidth $k$, then we can triangulate
it so that the resulting graph has pathwidth $O(k)$ (where the factors
are 1, 8 and 16 for 3-connected, 2-connected and arbitrary graphs).
With similar techniques, we also show that any outer-planar
graph of pathwidth $k$ can be turned into a maximal outer-planar
graph of pathwidth at most $4k+4$.  The previously best known
result here was $16k+15$.
\end{abstract}

\section{Introduction}

Let $G=(V,E)$ be an undirected simple graph that is  {\em planar}, i.e., it
has a crossing-free drawing in the plane.  
$G$ is called {\em triangulated} if all maximal regions not containing the
drawing are incident to three edges of $G$.  (More detailed definitions
will be given in Section~\ref{se:definition}.)
Any planar simple graph with $n\geq 3$ vertices can be triangulated
by adding edges without destroying planarity.

In this paper, we study the problem of triangulating a planar graph $G$
such that the pathwidth of the resulting graph is proportional
to the pathwidth of $G$.  Here, the {\em pathwidth} $pw(G)$ of a graph $G$ is
a well-known graph parameter (defined formally in Section~\ref{se:definitions}).
Graphs of small pathwidth have many applications.  Many graph problems
can be solved in polynomial  time if the pathwidth is constant. 
(See e.g.~\cite{Bod97}.)  The pathwidth also serves as
lower bound on the height of planar graph drawings \cite{FLW03}.  Vice versa,
some planar graphs $G$ can
be drawn with height $O(pw(G))$, notably trees \cite{Sud04}
and 2-connected outer-planar graphs \cite{Bie-WAOA12}.

The latter paper raised the question whether any outer-planar graph
can be made 2-connected by adding edges without increasing the
pathwidth much.  (For if so, then {\em all} outer-planar graphs
can be drawn with height $O(pw(G))$.)  This question was answered
in the affirmative by Babu et al.~\cite{BBC+14}, who showed that
any outer-planar graph $G$ can be made into a 2-connected outer-planar
graph $G'$ with $pw(G')\leq 16pw(G)+15$.

\paragraph{Our results: }
In this paper, we improve on the result by Babu et al.~and show
that we can add edges to any outer-planar graph $G$ such that the
result is a 2-connected outer-planar
graph $G'$ with $pw(G')\leq 4pw(G)+1$.  But our technique is much more
general.  Rather than working with outer-planar graphs, we
prove that any planar 2-connected graph can be triangulated without 
increasing the pathwidth if we allow multi-edges.  We can also
remove multi-edges; this increases the pathwidth at most 8-fold.
With much the same technique we can also handle graphs with cut-vertices
and make them 2-connected while increasing the pathwidth (roughly) 16-fold.
Outer-planar graphs can be handled as special cases and give an even
smaller increase in the pathwidth.

\paragraph{Related results: }  Many papers have dealt with
how to triangulate a planar graph under some additional constraint.
For example, any 2-connected planar graph can be 
triangulated so that the result is 4-connected (except for
wheel-graphs) \cite{BKK97}.  Any $k$-outer-planar graph can be
triangulated so that the result is $(k+1)$-outer-planar \cite{Bie15}.
Any planar graph $G$ with treewidth $tw(G)$ can be triangulated so that
the result has treewidth $\max\{3,tw(G)\}$ \cite{BR13}.
Triangulating planar graphs has also been studied while 
minimizing the maximum degree \cite{KB92}, and relates to planar graph
connectivity-augmentation problems (see e.g.~\cite{GMZ09} and the
references therein) since any triangulated graph 
is 3-connected.

\section{Background}
\label{se:definition}
\label{se:definitions}

Let $G=(V,E)$ be a graph with at least 3 vertices.
$G$ is called {\em planar} if it can be drawn
without crossing in the plane.
A crossing-free drawing $\Gamma$ of $G$ 
defines a cyclic order of edges at a vertex $v$
by enumerating them in clockwise order around $v$;
we call such a set of orders
a {\em planar embedding} of $G$.  The maximal
regions of $\mathbb{R}^2-\Gamma$ are called {\em faces} of the drawing; they
can be read from the planar embedding by computing the {\em
facial circuit}, i.e., the order of vertices and edges encountered while
walking around the face in clockwise order.
A graph $G$ is called {\em outer-planar} if $G\cup \{z^*\}$ is planar,
where $z^*$ is a newly-added {\em universal vertex} adjacent to all
vertices of $G$.

A {\em loop} is an edge $(v,v)$ for some vertex.  A {\em multi-edge}
is an edge $(v,w)$ with {\em multiplicity} $\mu \geq 2$, i.e., there
exist $\mu$ copies of $(v,w)$.  A graph is called {\em simple} if
it has neither loops nor multi-edges.  All input graphs in this
paper are required to be simple, but we sometimes add 
multi-edges in intermediate steps.  (We never add loops.)  A {\em multi-graph}
is a graph without loops (but possibly with multi-edges).  The {\em underlying
simple graph} of a multi-graph is obtained by deleting all but one copy
of each multi-edge.

\paragraph{Connectivity: }  A multi-graph $G$ is called
{\em connected} if we can go from any vertex $v$ to any vertex $w$
while walking along edges of $G$.  The {\em connected components}
of a multi-graph are the maximal subgraphs that are connected.
A multi-graph $G$ is called {\em $k$-connected}
if it remains connected even after deleting $k-1$ arbitrary vertices.
If $G$ is connected but not 2-connected, then $G$ has a {\em cut-vertex},
i.e., a vertex $v$ such that $G-v$ is not connected.  
A graph that is 2-connected, but not 3-connected, has a {\em cutting pair},
i.e., a pair of vertices $v,w$ such that $G-\{v,w\}$ is not connected.  

If $S$ is a set of vertices, then let $C'_1,\dots,C'_L$ be the connected
components of $G-S$ ($L=1$ if $S$ was not a cut-set). Define
for $i=1,\dots,L$ the {\em cut-component} $C_i$ of $S$ to consist of $C'_i$,
the edges from $C'_i$ to $S$, and a complete graph added between the
vertices of $S$.  
Define $\int C_i := C_i' = C_i-S$ to be the {\em interior} of $C_i$.

\paragraph{Triangulating: }
A face (in a planar graph in some planar embedding)
is called a {\em triangle} if its facial circuit contains three
edges.  A multi-graph $G$ is called {\em multi-triangulated} if 
it has a planar embedding such that all faces of $G$ are triangles.
Such a graph may well have multi-edges, but duplicate copies of an
edge must use different routes (no facial circuit
may consist of two copies of the same edge). 
A graph $G$ is called {\em triangulated} if it is multi-triangulated and simple.
A triangulated graph is 3-connected (and hence has a unique planar
embedding, up to reversal of all edge orders).  A multi-triangulated
graph $G$ need not be 3-connected, but it is 2-connected since $n\geq 3$
and $G$ has no loops.  One can show (see
\cite{Bie-WG15TR}) that the cutting pairs of $G$
correspond to multi-edges
as follows:  $\{u,v\}$ is a cutting pair that has $L$ cut-components
if and only if $(u,v)$ is a multi-edge with multiplicity $L$.
Further, $G$ has at least one edge that is not a multi-edge.

The idea of {\em triangulating} is to add edges to a graph until it is
triangulated.  More formally, {\em multi-triangulating a planar multi-graph $G$} 
means adding edges to $G$ so that the result is multi-triangulated.  
{\em Triangulating} a planar multi-graph $G$ means to add edges to
the underlying simple graph of $G$ such that the result is triangulated.
In particular, this operation is allowed to delete copies of a multi-edge
from $G$.

\medskip\noindent{\bf Pathwidth: }
Let $G$ be a multi-graph. 
Let $X_1,\dots,X_N$ be sets of vertices of $G$; we call these {\em bags}.  
We say that $X_1,\dots,X_N$ is a {\em path decomposition} $\calP$ of $G$ if
\begin{itemize}
\item every vertex appears in at least one bag,
\item for every edge $(u,v)$ in $G$, at least one bag $X_i$ contains
	both $u$ and $v$, and
\item for every vertex $v$ in $G$, the bags containing $v$
	form an interval.  Put differently, if $v\in X_{i_1}$ and
	$v\in X_{i_2}$ then also $v\in X_i$ for all $i_1<i<i_2$.
\end{itemize}
Bags naturally imply an order;
we write $X_i \preceq X_j$ if $i\leq j$ and $X_i \prec X_j$ if $i<j$.
The {\em bag-size} of such a path decomposition is $\max |X_i|$.
The {\em width} of such a path decomposition is $\max |X_i|-1$.
A graph is
said to have {\em pathwidth} at most $k$ if it has a path decomposition
of width $k$.

\section{3-connected graphs}
\label{se:connected}
\label{se:triangulating}
\label{se:triangulate}

We first show how to multi-triangulate 2-connected graphs (which
also triangulates 3-connected graphs).

\begin{lemma}
\label{lem:2connectedToTriangulatedMulti}
\label{lem:2conn_multi}
Let $G$ be a planar 2-connected multi-graph with a planar embedding 
for which any facial circuit has
at least 3 edges.  Then we can multi-triangulate $G$ without
increasing the pathwidth and without changing the planar embedding.
\end{lemma}
\begin{proof} 
\footnote{Babu et al.~published a similar proof in an early version of
\cite{BBC+14}, but omitted it in \cite{BBC+14}.}
Fix a path decomposition $\calP$ of $G$ that has width $pw(G)$.
Let $G^+$ be the graph induced by $\calP$, i.e., $G^+$
has the same vertices as $G$, but an edge $(v,w)$ for {\em any}
pair of vertices that occur in a common bag.  By properties of a
path decomposition
$G^+$ is an interval-graph, therefore chordal, therefore any simple cycle $C$
of length $\geq 4$ has a {\em chord} (an edge between two
non-consecutive vertices of $C$).  See Golumbic~\cite{Gol80} for details
of these concepts.  

Let $f$ be any facial circuit of $G$ with 4 or more edges on it. 
By 2-connectivity $f$ is a simple cycle, and hence $G^+$ contains a
chord of $C$.  Add this chord to $G$, routing it inside $f$.  The
resulting graph is still planar and 2-connected and all facial circuits
have at least 3 edges,
so repeat until $G$ is multi-triangulated.
\qed
\end{proof}

Our problem was motivated by planar graph drawing applications, where
often one starts by triangulating the planar graph (or adding edges
to the outer-planar graph to make it maximal outer-planar).  For
these applications, multi-edges are a problem. For example usually
one triangulates so that one can use the canonical
ordering \cite{FPP90} or a Schnyder wood \cite{Sch90}, and these only
exist for {\em simple} triangulated planar graphs.  
Hence
one wonders whether the same lemma holds without allowing multi-edges.
Thus, given a planar 2-connected graph, can we triangulate it without
increasing the pathwidth?  This turns out to be false.  Consider
a 4-cycle, which has pathwidth 2.  The only way to triangulate a
4-cycle without multi-edges is to turn it in $K_4$, which has
pathwidth 3.  However, if $G$ was already 3-connected, then no
multi-edges will happen.

\begin{corollary}
\label{cor:3connectedToTriangulatedSimple}
\label{cor:3conn}
Let $G$ be a 3-connected simple planar graph with $n\geq 3$.  
Then we can triangulate $G$ without increasing the pathwidth.
\end{corollary}
\begin{proof} 
Since $G$ is simple, any face has at least 3 edges.
Apply the previous lemma to get $G'$.  Adding edges cannot
decrease connectivity, so $G'$ has no cutting pairs.  
Since multi-edges in multi-triangulated graphs correspond to
cutting pairs, hence $G'$ is simple.
\qed
\end{proof}

\section{2-connected graphs}
\label{se:multi_edges}

We already know how to multi-triangulate 2-connected planar graphs
with Lemma~\ref{lem:2connectedToTriangulatedMulti}.
The hard part, done in this section, is how to convert such a
multi-triangulated graph into a triangulated one (i.e., remove
the multi-edges and replace them with others)
without increasing the pathwidth much.  
We state the required increase in terms of another parameter, $c$,
because this will help to obtain a smaller bound for outer-planar
graphs later.

\begin{lemma}
\label{lem:multi-edge}
\label{le:multi-edge}
\label{lem:multi}
Any multi-triangulated graph $G$ can be triangulated, after possibly
changing the planar embedding, such that the
resulting graph $G'$ has pathwidth $pw(G')\leq 2pw(G)+1+2c$.

Here $c$ is the maximum number of cutting pairs 
that can exist in one bag, i.e., for any path decomposition
$\calP$ of width $pw(G)$ and any bag $X_i$ of $\calP$ there are 
at most $c$ cutting pairs $\{u_1,v_1\},\dots, \{u_c,v_c\}$ 
such that $\{u_1,v_1,\dots,u_c,v_c\}\subseteq X_i$.
\end{lemma}

The rest of this section is devoted to the proof of this lemma.
We first give an outline of the proof.  We add $|X_i|+2c$ ``tokens'' to 
each bag $X_i$ of $\calP$; these are place-holders for vertices that need to be
added to bags later when adding edges.  These tokens are then
redistributed so that in each bag $X_i$ we have 2 tokens per cutting pair
$\{u,v\}\subseteq X_i$, and one token for each cut-component of $\{u,v\}$
that ``intersects'' $X_i$ in some sense.  We then can read from the
path-decomposition how to re-arrange the planar embedding such that
we can replace a copy of a multi-edge by a new edge in such a way that
we use up only ``few'' tokens. In particular, the above invariant on
what tokens exist in bags continues to hold.  Repeating this until no multi-edges
are left  then gives the desired graph $G'$.  Since we had $|X_i|+2c$
tokens, the new bag-size is at most $2|X_i|+2c$, and hence
$pw(G')\leq (2(pw(G)+1)+2c)-1 = 2pw(G)+1+2c$.

For the detailed proof, fix one planar embedding
of $G$ such that all faces are triangles. (We later
change this embedding, but all faces will continue to be triangles.)  
Fix one path decomposition $\calP$ of $G$ of width $pw(G)$.

\paragraph{Assigning tokens: }
We assign tokens to a bag $X_i$ of $\calP$ as follows:
(1) Add one token to $X_i$ for each vertex $v$ in $X_i$; 
	this is the {\em vertex-token} of $v$.
(2) Add two tokens to $X_i$ for every cutting pair $\{u,v\}$ 
	with $\{u,v\}\subseteq X_i$; these are the
	{\em cutting-pair tokens}, or the {\em tokens of $\{u,v\}$}.

\paragraph{Peripheral pairs: }
Let $\{u,v\}$ be a cutting pair, and let $C_0,\dots,C_L$ be its
cut components.    One can show \cite{Bie-WG15TR}
that for $i\in\{0,\dots,L\}$
the edges from $v$ to $\int C_i$ occur 
consecutively in the clockwise order of edges around $v$,
surrounded by two copies of edge $(u,v)$.
See Figure~\ref{fig:peripheral} for an illustration.
Let $b_i^\ell$ and $b_i^r$ be the first and last neighbor of $v$
within this interval of edges to $\int(C_i)$.
We call $\{b_i^\ell,b_i^r\}$
the {\em peripheral pair of cut-component $C_i$}.
Notice that $b_i^\ell=b_i^r$ if $\deg(b_i^\ell)=2$ or
$(b_i^\ell,v)$ is a multi-edge,
but we use the term ``pair'' even then for ease of wording.

\vspace*{-5mm}
\begin{figure}[ht]
\hspace*{\fill}
\includegraphics[width=0.47\linewidth,page=3]{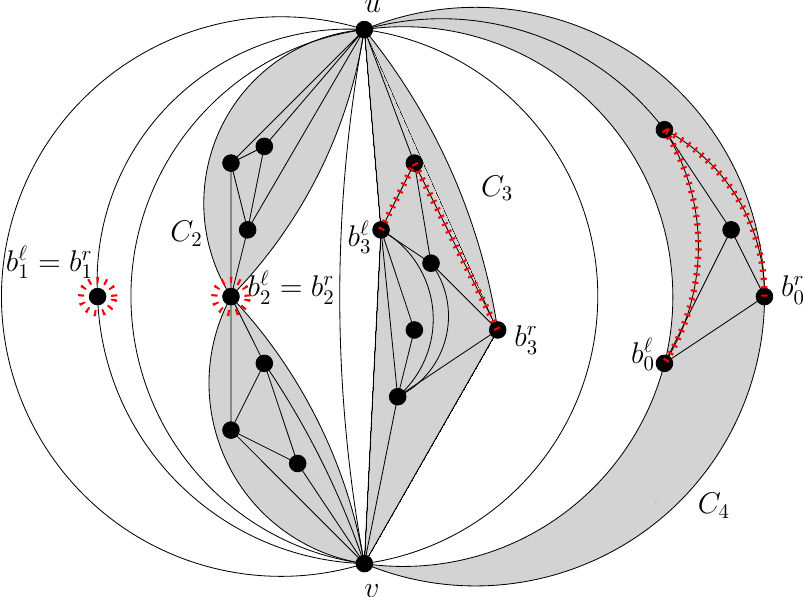}
\hspace*{\fill}
\includegraphics[width=0.47\linewidth,page=2]{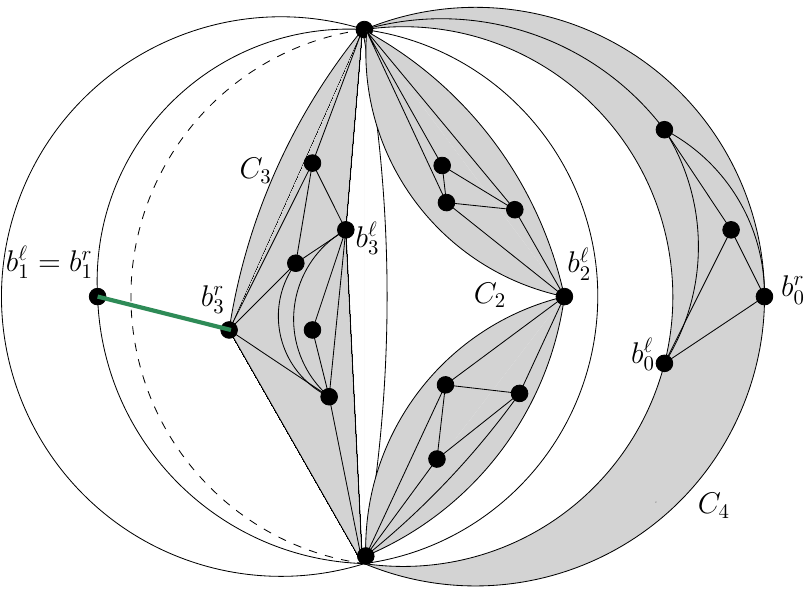}
\hspace*{\fill}
\caption{A multi-triangulated graph with a cutting pair $\{u,v\}$
that has four cut-components. Dotted red lines are paths
assigned to peripheral pairs as in Lemma~\ref{lem:paths}.  
We can add edge $(b_1^r,b_3^r)$ 
if we swap $C_2$ and $C_3$ and reverse $C_3$.}
\label{fig:peripheral}
\end{figure}

\vspace*{-5mm}
\begin{observation}
\label{obs:planarity}
\label{obs:planar}
Let $G$ be a multi-triangulated graph that has a cutting pair $\{u,v\}$.
Let $C_i$ and $C_j$ be two different cut-components of 
$\{u,v\}$.  For any choice of $\alpha,\beta\in \{\ell,r\}$, 
deleting one copy of $(u,v)$ and adding
$(b_i^\alpha,b_j^\beta)$ results in a multi-triangulated graph
(after possibly changing the planar embedding).
\end{observation}
\begin{proof}
This follows from the results in \cite{DiBattista96b}. In a nutshell, we
can reverse and swap cut-components until $b_i^\alpha$ and $b_j^\beta$
both face one copy of $(u,v)$.  Deleting this copy gives a face with 4
edges; inserting  edge $(b_i^\alpha,b_j^\beta)$ into this face gives a
planar graph where all faces are triangles.
\qed
\end{proof}

\paragraph{Bag-intervals: }
Let $\{b_i^\ell,b_i^r\}$ be the peripheral-pair of a cut-component $C_i$
of a cutting pair $\{u,v\}$.  Since $G$ is multi-triangulated, $\{u,v,
b_i^\alpha\}$ forms a triangle for $\alpha\in \{\ell,r\}$.
By the properties of the path decomposition there must exist at least one bag 
that contains all three vertices.  Thus let $X(b_i^\alpha)$ be a
bag containing $\{u,v,b_i^\alpha\}$; choose an arbitrary one if there is
more than one.
So far the superscripts $\ell$ and $r$ for $\{b_i^\ell,b_i^r\}$ effectively 
meant ``one'' and ``the other'', since we can reverse the planar embedding 
of cut-component $C_i$.  We now fix the superscripts
such that $X(b_i^\ell) \preceq X(b_i^r)$, i.e., the bag of $b_i^\ell$
is left of the bag of $b_i^r$.  The left-open set of bags 
$(X(b_i^\ell),X(b_i^r)]:=
\{X: X(b_i^\ell)\prec X \preceq X(b_i^r)\}$ is called 
the {\em bag-interval} of peripheral pair $\{b_i^\ell,b_i^r\}$.
Notice that the bag-interval is empty if $X(b_i^\ell)=X(b_i^r)$;
this will not pose problems.

\paragraph{Child-peripheral-pairs: }
So far all cut-components at a cutting pair have been treated equally.
For token-accounting-purposes, we introduce a hierarchy among them.
Fix one edge $e$ of $G$ that is not a multi-edge. 
For each cutting pair $\{u,v\}$
with cut-components $C_0,\dots,C_L$, the {\em parent-component}
of $\{u,v\}$ is the one that contains edge $e$, while all other
cut-components are called {\em child-components}.  Correspondingly
we call a peripheral-pair of $\{u,v\}$ a {\em child-peripheral-pair}
if it belongs to a child-component of $\{u,v\}$.  

\paragraph{Redistributing tokens: }
Let $\calB$ be the union, over all cutting pairs $\{u,v\}$, of all
the child-peripheral-pairs of $\{u,v\}$.  We want to redistribute
vertex-tokens to child-peripheral-pairs, and for this we need an observation.

\begin{lemma}
\label{lem:paths}
Let $\calB$ be the set of all child-peripheral pairs in $G$.
There exists a set of vertex-disjoint paths $P_1,\dots,P_{|\calB|}$ in $G$ 
such that for any child-peripheral-pair $\{b^\ell,b^r\}$ in $\calB$, one of the
paths connects $b^\ell$ with $b^r$.
\end{lemma}
\begin{proof} (Sketch) 
Consider any child-peripheral-pair $\{b_i^\ell,b_i^r\}$,
say at cut-component $C_i$ of cutting pair $\{u,v\}$.  Observe that 
there are three vertex-disjoint paths from $b_i^\ell$ to $b_i^r$:
one via $u$, one via $v$, and one within $\int(C_i)=C_i-\{u,v\}$ 
since the latter
is connected by definition of cut-components.  Since $(u,v)$ is an edge,
therefore $\{u,v,b_i^\ell,b_i^r\}$ must all
belong to one triconnected component, call it $D$.
Since $D$ is 3-connected, there must exist a path in $D-\{u,v\}$
connecting $b_i^\ell$ and $b_i^r$.
One can now show (see \cite{Bie-WG15TR}) that choosing this path for 
peripheral pair $\{b_i^\ell,b_i^r\}$ will assign vertex-disjoint paths
to all child-peripheral pairs.
\qed
\end{proof}

We now redistribute vertex-tokens to child-peripheral pairs as follows.  
For every child-peripheral-pair $\{b^\ell,b^r\}$, find the path $P$ connecting
$b^\ell$ and $b^r$ from Lemma~\ref{lem:paths}.  For every vertex $w\in P$,
declare the vertex-token of $w$ to belong to the
child-peripheral-pair $\{b^\ell,b^r\}$; we now call it a
{\em child-peripheral-pair token} and say that it {\em belongs to 
$\{b^\ell,b^r\}$}.  Since the paths of 
child-peripheral-pairs are vertex-disjoint, every vertex-token is used
at most once.

By properties of a path decomposition, the set of bags 
$\calX_P=\{X: X$ contains a vertex of $P\}$ forms an interval of bags
since $P$ is connected.  Each bag in $\calX_P$ obtains at least one
token of $\{b^\ell,b^r\}$.  Since $X(b^\ell),
X(b^r)\in \calX_P$, we therefore have:

\begin{invariant}
\label{inv}
(1) For every child-peripheral-pair $\{b^\ell,b^r\}$, every
bag $X$ in the bag-interval $(X(b^\ell),X(b^r)]$ contains
at least one token of $\{b^\ell,b^r\}$.
(2)
For every cutting pair $\{u,v\}$, every bag containing
both $u$ and $v$ contains two tokens of $\{u,v\}$.
\end{invariant}

\paragraph{Adding edges: } 
We now repeatedly delete one copy of a multi-edge $(u,v)$
and replace it with some edge $(b_i^\alpha,b_j^\beta)$
between two different cut-components of $\{u,v\}$.  Notice
that no such edge can have existed before, so the sum of
the multiplicities of multi-edges decreases.  By
Observation~\ref{obs:planar}, adding these edges maintains
a multi-triangulation.  After
repeated applications we hence end with a simple graph.  
Throughout these edge additions, we maintain a valid path
decomposition for the graph by adding vertices to bags, if
needed.  This uses up some tokens, but we do it in such a
way that Invariant~\ref{inv} is maintained and hence the
pathwidth is at most $2pw(G)+1+2c$.

So let $\{u,v\}$ be a cutting pair.  
Let $C_0,\dots,C_L$ be the cut components of $\{u,v\}$, with
$C_0$ the parent-component.  For each component $C_i$,
let $\{b_i^\ell,b_i^r\}$ be the peripheral-pair of $C_i$.  
We distinguish cases.
\begin{enumerate}
\item There exists some $i\neq j$, $i>0$, $j>0$
	such that $X(b_i^\ell) \prec X(b_j^\ell) \prec  X(b_i^r) 
	\prec X(b_j^r)$.
	Put differently, there are two child components  $C_i$ and $C_j$ 
	whose bag-intervals intersect, 
	but neither one contains the other.
	See also Figure~\ref{fig:replace_tokens_1}.

	Add an edge $(b_j^\ell,b_i^r)$.
	Since both $C_i$ and $C_j$ are child-components, 
	by the invariant each bag $X$ with $X(b_j^\ell)\prec X \preceq X(b_i^r)$
	contains one token of
	$\{b_j^\ell,b_j^r\}$ and one token of
	$\{b_i^\ell,b_i^r\}$.  We use one of them to add
	$b_j^\ell$ to all these bags; then $b_j^\ell$ and
	$b_i^r$ share a bag, the bags containing $b_j^\ell$ continue
	to form an interval, and we  hence have a valid
	path decomposition for the new graph. 

	Adding the edge combines child-components $C_i$ and $C_j$ into
	one new child-component $C'$  with peripheral-pair
	$\{b_i^\ell,b_j^r\}$.  Since we used only one token in each bag, all
	bags $X$ with $X(b_i^\ell) \prec X \preceq X(b_j^r)$ have
	a peripheral-pair-token left, which we now assign to $C'$.
	So the invariant holds. 
\item There exists some $i\neq j$, $i>0$, $j>0$,
	such that $X(b_i^\ell) \preceq X(b_j^\ell) \preceq X(b_j^r) \preceq X(b_i^r)$.
	Put differently, there are two child components  $C_i$ and $C_j$ 
	whose bag-intervals intersect,  and one
	is inside the other.	
	See also Figure~\ref{fig:replace_tokens_1}.

	Add an edge $(b_i^\ell,b_j^\ell)$.
	Each bag $X$ with $X(b_i^\ell)\prec X \preceq X(b_j^\ell)$
	contains a token of $\{b_i^\ell,b_i^r\}$.
	We use this to add $b_i^\ell$ to all these bags; then $b_i^\ell$ and
	$b_j^\ell$ share a bag and the bags containing $b_i^\ell$ are
	consecutive, hence we have a valid
	path decomposition of the new graph.

	Adding the edge combines components $C_i$ and $C_j$ into
	one new component $C'$ with
	peripheral-pair $\{b_i^r,b_j^r\}$.  Since we used only tokens in bags
	farther to the left, all
	bags $X$ with $X(b_j^r) \prec X \preceq X(b_i^r)$ still have
	the token of $\{b_i^\ell,b_i^r$\}, and
	we assign these to the new peripheral-pair.
	So the invariant holds.

\newpage
\begin{figure}[t]
\hspace*{\fill}
\includegraphics[width=0.6\linewidth,page=1]{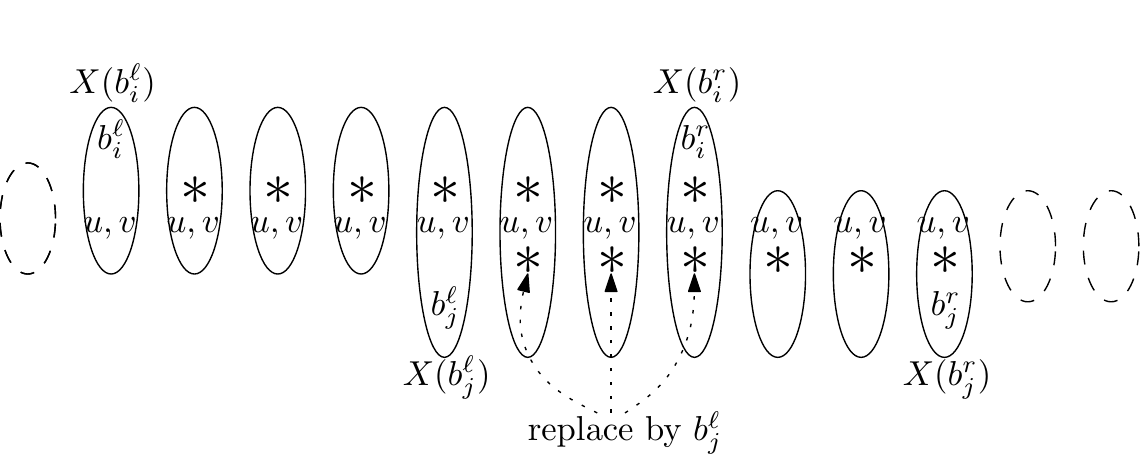}
\hspace*{\fill} \\
\hspace*{\fill}
\includegraphics[width=0.6\linewidth,page=2]{replaceTokens.pdf}
\hspace*{\fill}
\caption{Bag-intervals with peripheral-pair-tokens (shown with $*$). 
(Top) The bag-intervals intersect, but neither contains the
other.  (Bottom) One bag-interval is a subset of the other. }
\label{fig:replace_tokens_1}
\end{figure}

\vspace*{-10mm}
\item No two bag-intervals of two child-components
	intersect.  After possible renaming of the child components	
	$C_1,\dots,C_L$, we may hence assume that
	$X(b_1^\ell)\preceq X(b_1^r)\preceq X(b_2^\ell)\preceq X(b_2^r)\preceq 
	\dots \preceq  X(b_L^\ell)\preceq X(b_L^r).$
	(The bag-interval of the parent-component may be anywhere in this 
	order.)
	See also Figure~\ref{fig:replace_tokens_2}.

\begin{figure}[ht]
\hspace*{\fill}
\includegraphics[width=0.7\linewidth,page=3]{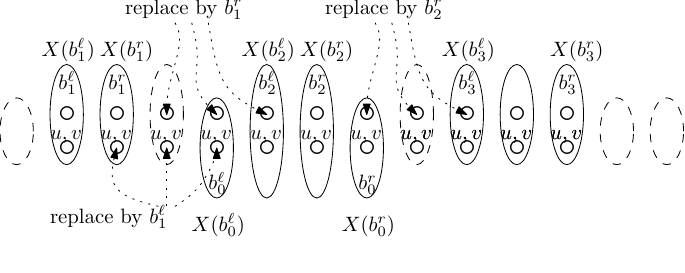}
\hspace*{\fill}
\caption{Replacing cutting-pair-tokens (shown with $\circ$) to combine all
remaining cut-components of cutting pair $\{u,v\}$.  }
\label{fig:replace_tokens_2}
\end{figure}

	We will combine {\em all} cut components
	into one at once. 
	Add edges $(b_1^r,b_2^\ell)$,
	$(b_2^r,b_3^\ell),\dots,(b_{L-1}^r,b_L^\ell)$.  To create
	a path decomposition for this, add $b_i^r$ to all bags
	$X$ with $X(b_i^r)\prec X \preceq X(b_{i+1}^\ell)$, for $i=1,\dots,L-1$.  
	Pay for these additions with the first token of $(u,v)$.
	We know that each of these bag has such a token, since $X(b_1^\ell)$
	and $X(b_L^r)$ contain $\{u,v\}$ by definition, and the bags
	between must contain $\{u,v\}$ by properties of a path
	decomposition.

	Finally add edge $(b_1^\ell,b_0^r)$.  Create a path decomposition
	for this
	by adding $b_1^\ell$ to all bags from $X(b_1^\ell)$ to
	$X(b_0^r)$,
	and pay for it with the second cutting-pair-token of $(u,v)$.

	Observation~\ref{obs:planar} applies to all added
	edges, since the ends of each edge are peripheral-vertices of
	two different cut-components, even after considering that previous
	edge-additions merged some them.  Hence the
	resulting graph is a multi-triangulation after we deleted $L+1$
	copies of multi-edge $(u,v)$.

	Since $\{u,v\}$ ceases to be a
	cutting pair after adding these edges, the invariant holds again
	since we only used tokens of $\{u,v\}$.
\end{enumerate}
After repeatedly applying the above edge-additions to all cutting
pairs, we hence end with a triangulated graph and a path decomposition
of width at most $2pw(G)+1+2c$ as desired.
This proves Lemma~\ref{lem:multi}.  

\bigskip

\begin{lemma}
\label{lem:2conn}
\label{lem:2conn_simple}
Let $G$ be a 2-connected planar graph with $n\geq 3$ vertices.  
Then we can triangulate
$G$, after possibly changing the planar embedding, such that the 
result has pathwidth at most $8pw(G)-5$.
\end{lemma}
\begin{proof}
By Lemma~\ref{lem:2connectedToTriangulatedMulti}
we ca multi-triangulate $G$ without increasing
the pathwidth.  Call the result $G_1$.
By Lemma~\ref{lem:multi} we can triangulate $G_1$
such that the resulting graph $G_2$ has $pw(G_2)\leq 2pw(G)+1+2c$.

It remains to bound $c$.  Recall that this is the maximum number of
cutting pairs of $G_1$ for which all vertices occur in one bag $X_i$
(of some path decomposition $\calP$ of width $pw(G_1)$).
Each such cutting pair corresponds to a multi-edge in $G_1$.
Let $G[X_i]$ be the graph induced by $X_i$ and $G_s$ be its
underlying simple graph.  Each such cutting pair hence corresponds
to an edge in $G_s$.  Since $G_s$ is planar and simple and has $|X_i|$
vertices, it has at most $3|X_i|-6\leq 3(pw(G)+1)-6=3pw(G)-3$ edges 
if $|X_i|\geq 3$.  If $|X_i|\leq 2$, then $G_s$ has at most
$1\leq 3pw(G)-3$ edges since $pw(G)\geq 2$ (a graph of pathwidth 1
is a forest and cannot be 2-connected).
Thus  either way $G_s$ has at most $3pw(G)-3$ edges, hence $c\leq 3pw(G)-3$
and $pw(G_2)\leq 2pw(G)+1+2c \leq 
8pw(G)-5$ as desired.
\qed
\end{proof}

\section{2-connecting an outer-planar graph}

Recall that one motivation for this paper was the question how
to make an outer-planar graph 2-connected by adding edges without
increasing the pathwidth much.  A {\em maximal outer-planar graph}
is a simple outer-planar graph to which we cannot add edges without
violating planarity, simplicity, or outer-planarity.  Such a graph
is 2-connected.

\begin{theorem}
Let $G$ be a simple connected outer-planar graph.
Then we can add edges to $G$, after possibly changing the planar
embedding,
to obtain a maximal outer-planar graph $G'$ with $pw(G')\leq 4pw(G)+4$.
\end{theorem}
\begin{proof}
If $n=1$ then $G$ is already maximal outer-planar, so assume $n\geq 2$.
Add a universal vertex $z^*$ to $G$ and call the result $G_1$; we
know that $G_1$ is planar and $pw(G_1)= pw(G)+1$ since we can
add $z^*$ to all bags. 
Observe that $G_1-v$ is connected for any $v\neq z^*$ since $z^*$ is
adjacent to all vertices.  Therefore $G_1$ is 2-connected and
any cutting pair of $G_1$ must include $z^*$.

Use Lemma~\ref{lem:2conn_multi} to multi-triangulate $G_1$ without
increasing pathwidth, and call the result $G_2$; we have
$pw(G_2)=pw(G)+1$.  Now use Lemma~\ref{le:multi-edge} to 
triangulate $G_2$, and call the result $G_3$.  We have
$pw(G_3)\leq 2pw(G_2)+1+2c \leq 2pw(G)+3+2c$.

Since any cutting pair includes $z^*$, we can get an improved bound for $c$
as follows.
Let $\calP_2$ be any path decomposition of $G_2$ of width $pw(G_2)$
and let $X_i$ be any bag of $\calP_2$; we have $|X_i|\leq pw(G_2)+1
=pw(G)+2$.  If $X_i$ contains cutting pairs, then it must contain
$z^*$.  Each such cutting pair uses $z^*$ and one other vertex in
$X_i$, so there are at most $|X_i|-1$ cutting pairs with both ends
in $X_i$, and $c\leq |X_i|-1 \leq pw(G)+1$.
Putting it all together, we have
$pw(G_3)\leq 2pw(G)+3+2(pw(G)+1)=4pw(G)+5$.

Finally delete the added vertex $z^*$ to obtain $G_4$, which has the same
vertices as $G$.  Since $z^*$ was universal and
$G_3$ was triangulated, $G_4$ is maximal
outer-planar.  Since $z^*$ was universal,
$pw(G_4)=pw(G_3)-1\leq 4pw(G)+4$ and hence $G_4$ satisfies all conditions on $G'$.
\qed
\end{proof}

We note here that the bound can be improved to $4pw(G)+3$ by delving
into the proofs of Lemma~\ref{lem:multi} and
Lemma~\ref{lem:paths} and observing that the vertex-token of $z^*$
will never be used as child-peripheral-pair-token, since $z^*$ is in
all cutting pairs. We leave the details to the reader.

\section{All graphs}
\label{se:cutvertices}
\label{sec:cutvertices}

We now show how to handle cutvertices and disconnected graphs.

\begin{lemma}
\label{lem:conn}
Any simple connected planar graph $G$ with $n\geq 3$ can
be triangulated, after possibly changing the planar embedding,
so that the result has pathwidth at most $16pw(G)+3$.
\end{lemma}
\begin{proof}
Let $v_1$ be a cut-vertex of $G$. Add a new vertex $z_1$ as follows.
\smallskip

\begin{minipage}{0.4\linewidth}
\begin{picture}(0,0)%
\includegraphics{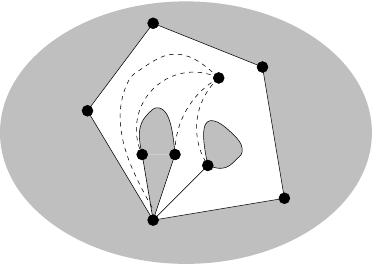}%
\end{picture}%
\setlength{\unitlength}{1381sp}%
\begingroup\makeatletter\ifx\SetFigFont\undefined%
\gdef\SetFigFont#1#2#3#4#5{%
  \reset@font\fontsize{#1}{#2pt}%
  \fontfamily{#3}\fontseries{#4}\fontshape{#5}%
  \selectfont}%
\fi\endgroup%
\begin{picture}(5100,3598)(301,-9660)
\put(2551,-9361){\makebox(0,0)[lb]{\smash{{\SetFigFont{7}{8.4}{\rmdefault}{\mddefault}{\updefault}{\color[rgb]{0,0,0}$v_1$}%
}}}}
\put(3001,-6961){\makebox(0,0)[lb]{\smash{{\SetFigFont{7}{8.4}{\rmdefault}{\mddefault}{\updefault}{\color[rgb]{0,0,0}$z_1$}%
}}}}
\end{picture}%
\end{minipage}
\begin{minipage}{0.55\linewidth}
Let $C_0,\dots,C_L$ be the cut-components of $v_1$.  
Rearrange
the planar embedding at $v_1$ such that for each $C_j$ the edges from $v_1$
to $C_j$ are consecutive at $v_1$.  
In consequence, there
now exists a face $f_1$ that is incident to all cut-components
of $v_1$.  Insert a new vertex $z_1$ in face $f_1$, and make it
adjacent to $v_1$ and to all neighbors $x$ of $v_1$  that are
on $f_1$.  
Afterwards $v_1$ is no longer a cut-vertex, and $z_1$ is also
not a cut-vertex. 
\end{minipage}

\medskip
We can obtain a path decomposition of $G\cup \{z_1\}$ by taking
one of $G$ and adding
$z_1$ to all bags that contains $v_1$.  This covers all new edges
since all neighbors of $z_1$ are neighbors of $v_1$.

Repeat the process in the resulting graph until there are
no cut-vertices left.  Call the final graph $G_1$.  Since
none of the new vertices were cut-vertices, we added at most $|X_i|$
new vertices to each bag $X_i$ of a path decomposition of $G$.
Hence the bag-size at most doubles and
$pw(G_1)\leq 2pw(G)+1$.

Now multi-triangulate $G_1$ with
Lemma~\ref{lem:2conn_multi} and call the result $G_2$.
We have $pw(G_2)= pw(G_1) \leq 2pw(G)+1$.
Now triangulate $G_2$ with Lemma~\ref{lem:2conn_simple} and
call the result $G_3$. 
We have $pw(G_3)\leq 8pw(G_2)-5 \leq 8(2pw(G)+1)-5 = 16pw(G)+3$.

Now we must remove the added vertices while keeping a
triangulated graph, and do this by contracting each 
into a suitable neighbor.  Observe that the neighbors of
$z_1$ form a simple cycle since $G_3$ is triangulated. 
Hence these neighbors induce a simple outer-planar 2-connected
graph.  It is well-known that every such
graph has a vertex of degree 2.  Therefore $z_1$ has
a neighbor $y_1$ such that $y_1$ and $z_1$ have exactly two
common neighbors (which are the third vertices on the faces
incident to edge $(z_1,y_1)$).  {\em Contract} edge $(z_1,y_1)$,
i.e., delete $z_1$ and re-route every incident edge of $z_1$
to end at $y_1$ instead.  Delete resulting loops and 
multi-edges.  Because $z_1$ and $y_1$ had exactly two neighbors in
common, the resulting graph is again triangulated.
Repeat the process for the other added vertices.

At the end the graph $G_4$ that results
has the same vertices as $G$.  It is well-known that contraction
of an edge does not increase pathwidth, so
$pw(G_4)\leq pw(G_3)\leq  16pw(G)+3$ as desired.
\qed
\end{proof}

As for disconnected graphs, one can easily show the following \cite{Bie-WG15TR}:

\begin{lemma}
\label{lem:make_connected}
Let $G$ be a planar graph. 
Then we can add edges to $G$ so that
the resulting graph $G'$ is planar, connected, and $pw(G')=\max\{1,pw(G)\}$.
\end{lemma}

Hence we can triangulate $G$ by first creating $G'$ and then triangulating
$G'$.

\section{Conclusion}
\label{se:conclusion}

In this paper, we studied how to add edges to a planar graph without
increasing the pathwidth much.
We summarize all our results with the following:

\begin{theorem}
\label{th:triangulatedSimple}
Let $G$ be a simple planar 
graph with at least 3 vertices.  
Then we can triangulate $G$ such that the result $G'$ has
\begin{itemize}
\item $pw(G')= pw(G)$ if $G$ is 3-connected,
\item $pw(G')\leq 8pw(G)-5$ if $G$ is 2-connected,
\item $pw(G')\leq 16pw(G)+3$ otherwise.
\end{itemize}
\end{theorem}

It may also be of interest to observe that our construction does not
change a given path decomposition of the graph other than by adding 
more vertices to some bags.  On the other hand, our construction often
changes the planar embedding.  Is it possible to triangulate a graph
without increasing the pathwidth much and without changing the planar
embedding?

Following the steps of the proof, one can see that 
the triangulation can be found in linear time, presuming that
we are given a path decomposition of width $pw(G)$ in the form 
of the index of the first and last bag containing $v$ for every
vertex $v$.  There is no need to compute triconnected
components: One can find child-components via
multi-edges, and the paths in Lemma~\ref{lem:paths} are
only needed for accounting purposes and need not be computed.

The obvious open problem is to improve the factors, especially for 
2-connected graphs.  Can every planar graph $G$ be triangulated so
that the result has pathwidth at most $\max\{3,pw(G)\}$?  

It would also be of interest to study other
width-parameters (such as the carving width, bandwidth, 
clique-width, etc.) and ask whether planar graphs can be triangulated
while keeping the width-parameter asymptotically the same.

\bibliographystyle{plain}
\bibliography{../../bib/journal,../../bib/full,../../bib/gd,../../bib/papers,extra}

\begin{thebibliography}{10}

\bibitem{BBC+14}
J. Babu, M. Basavaraju, L.~Sunil Chandran, and D. Rajendraprasad.
\newblock 2-connecting outerplanar graphs without blowing up the pathwidth.
\newblock {\em Theor. Comput. Sci.}, 554:119--134, 2014.

\bibitem{BKK97}
\student{T. Biedl}, G.~Kant, and M.~Kaufmann.
\newblock On triangulating planar graphs under the four-connectivity
  constraint.
\newblock {\em Algorithmica}, 19(4):427--446, 1997.

\bibitem{Bie-WAOA12}
T.~Biedl.
\newblock A 4-approximation algorithm for the height of drawing 2-connected
  outerplanar graph.
\newblock In {\em Workshop on Approximation and Online Algorithms (WAOA'12)},
  volume 7846 of {\em Lecture Notes in Computer Science}, pages 272--285.
  Springer-Verlag, 2013.

\bibitem{Bie15}
T.~Biedl.
\newblock On triangulating $k$-outerplanar graphs.
\newblock {\em Discrete Applied Mathematics}, 181:275--279, 2015.

\bibitem{Bie-WG15TR}
T.~Biedl.
\newblock Triangulating planar graphs while keeping the pathwidth small.
\newblock Technical report, ArXiV, 2015.
\newblock {To appear}.

\bibitem{BR13}
T.~Biedl and \student{L.E. Ruiz Vel\'{a}zquez}.
\newblock Drawing planar 3-trees with given face areas.
\newblock {\em Computational Geometry: Theory and Applications},
  46(3):276--285, 2013.

\bibitem{Bod97}
H.~Bodlaender.
\newblock Treewidth: algorithmic techniques and results.
\newblock In {\em Mathematical Foundations of Computer Science (MFCS 1997)},
  volume 1295 of {\em Lecture Notes in Computer Science}, pages 19--36.
  Springer-Verlag, 1997.


\bibitem{DiBattista96b}
G.~Di Battista and R.~Tamassia.
\newblock On-line planarity testing.
\newblock {\em {SIAM} J. Computing}, 25(5), 1996.

\bibitem{FLW03}
S.~Felsner, G.~Liotta, and S.~Wismath.
\newblock Straight-line drawings on restricted integer grids in two and three
  dimensions.
\newblock {\em Journal of Graph Algorithms and Applications}, 7(4):335--362,
  2003.

\bibitem{FPP90}
H.~de Fraysseix, J.~Pach, and R.~Pollack.
\newblock How to draw a planar graph on a grid.
\newblock {\em Combinatorica}, 10:41--51, 1990.

\bibitem{Gol80}
M.~C. Golumbic.
\newblock {\em Algorithmic graph theory and perfect graphs}.
\newblock Academic Press, New York, 1st edition, 1980.

\bibitem{GMZ09}
C.~Gutwenger, P.~Mutzel, and B.~Zey.
\newblock On the hardness and approximability of planar biconnectivity
  augmentation.
\newblock In {\em Computing and Combinatorics (COCOON'09)}, volume 5609 of {\em
  LNCS}, pages 249--257. Springer, 2009.

\bibitem{KB92}
G.~Kant and H.~Bodlaender.
\newblock Triangulating planar graphs while minimizing the maximum degree.
\newblock In {\em Scandinavian Workshop on Algorithm Theory (SWAT'92)}, volume
  621 of {\em LNCS}, pages 258--271. Springer, 1992.

\bibitem{Sch90}
W.~Schnyder.
\newblock Embedding planar graphs on the grid.
\newblock In {\em {ACM}-{SIAM} Symposium on Discrete Algorithms (SODA '90)},
  pages 138--148, 1990.

\bibitem{Sud04}
M.~Suderman.
\newblock Pathwidth and layered drawings of trees.
\newblock {\em International Journal of Computational Geometry and
  Applications}, 14(3):203--225, 2004.

\end{thebibliography}

\newpage

\begin{appendix}
\section{Missing details}

\subsection{Properties of multi-triangulated graphs}

\begin{lemma}
\label{lem:mtriangulation}
Let $G$ be a multi-triangulated planar graph with $n\geq 3$ vertices.    
Fix an
arbitrary planar embedding for which all faces are triangles.
The following holds:
\begin{enumerate}
\item $G$ is 2-connected.
\item Any cutting pair $\{u,v\}$ gives rise to a multi-edge $(u,v)$.
\item For any multi-edge $(u,v)$, $\{u,v\}$ is a cutting pair,
	and the number of its cut-components equals the multiplicity
	of the multi-edge.
\item For any cutting pair $\{u,v\}$ with a cut-component $C$
	in the order of edges around $u$ the edges to $\int(C)$ appear
	consecutively, and are preceded and succeeded by copies of $(u,v)$.
\end{enumerate}
\end{lemma}
\begin{proof}
Let $S$ be a cut-set (i.e., either cut-vertex or cutting pair).
Consider a vertex $v\in S$.  Assume for
contradiction that in the clockwise order around $v$ there are two
consecutive neighbors $w_1,w_2$ with $w_1\in \int(C_1)$ and $w_2
\in \int(C_2)$ for two different cut-components $C_1,C_2$ of $S$.
Consider the face $f$ that is between edges $(v,w_1),(v,w_2)$ at $v$.
Since $w_1,w_2$ are in the interior of different cut-components, 
we cannot have an edge $(w_1,w_2)$.  We must have $w_1\neq v\neq w_2$,
since otherwise there would be a loop.  Therefore face $f$ is incident to at
least 4 edges.  Contradiction.

Thus for any two cut-components of $S$, edges from $v$ to the inside
of the cut-component cannot be consecutive.  Thus, there must an
edge between any two cut-components (in the clockwise order around $v$)
for which the other endpoint is also in $S$.  If $|S|=1$ then such
an edge would be a loop, a contradiction.  Therefore no cut-set can
have size $1$ and $G$ is 2-connected; this proves (1).  

If $|S|=2$, say $S$ is the cutting pair $\{u,v\}$, then the cut-components
are separated by copies of edge $(u,v)$.  If there are $L$ cut-components
$C_1,\dots,C_L$ for $L\geq 2$, then there are at least $L$ places in
the clockwise order around $v$ where we switch from one cut-component
to the next one, so we must have at least $L$ copies of $(u,v)$.  This
proves (2).

Let $e_0,\dots,e_{\ell-1}$ be the copies of $(u,v)$, enumerated 
in the clockwise order around $v$.  We have just shown $\ell\geq L$.
For $i=1,\dots,\ell$, edges
$e_{i-1}$ and $e_i$ cannot be consecutive at $v$ (where indices
are modulo $L$), otherwise there would be a face of
degree 2.  So there must be vertices other than $u$ between
$e_{i-1}$ and $e_i$.  Further, the cycle formed by $e_{i-1}$ and $e_i$
separates everything on one side from everything on the other side.
So the subgraph between $e_{i-1}$ and $e_i$ contains at least one
cut-component of $\{u,v\}$.  It follows that $\ell\leq L$, and
so $\ell=L$.  This proves (3).

Since $\ell=L$, the subgraph between $e_{i-1}$ and $e_i$ must contain
exactly one cut-component of $\{u,v\}$.  Therefore in the cyclic
order around $v$ we alternate between a copy of $(u,v)$ and all
edges to exactly one cut-component.  This proves (4).
\qed
\end{proof}

\begin{lemma}
Every multi-triangulated graph has at least one edge that is
not a multiple edge.
\end{lemma}
\begin{proof}
Fix one arbitrary planar drawing $\Gamma$ of $G$ for which all facial
circuits have three edges.  Nothing is to show if $G$ is
simple, so assume $G$ has multi-edges.  If $e_1,e_2$ are
two copies of a multi-edge, then their drawing defines a
closed curve $C$.  This curve cannot be the boundary of a
face since facial circuits have three edges.  In consequence,
at least one vertex must be inside any closed curve defined
by two copies of a multi-edge.

Assume that $e_1,e_2$ has been chosen such that their closed
curve encloses the minimum possible number of vertices among
all such pairs.  Let $v$ be a vertex inside that curve, and
let $e$ be an edge incident to $v$.  Then $e$ must be simple
by choice of $e_1,e_2$.
\qed
\end{proof}

\subsection{Finding paths for child-peripheral pairs}

This section gives the proof of Lemma~\ref{lem:paths}, which
states the following:
\begin{quotation}
Let $\calB$ be the set of all child-peripheral pairs in $G$.
There exists a set of vertex-disjoint paths $P_1,\dots,P_{|\calB|}$ in $G$ 
such that for any child-peripheral-pair $\{b^\ell,b^r\}$ in $\calB$, one of the
paths connects $b^\ell$ with $b^r$.
\end{quotation}
Consider any child-peripheral-pair $\{b_i^\ell,b_i^r\}$,
say at cut-component $C_i$ of cutting pair $\{u,v\}$.  As argued
in the main part of the paper, then $\{u,v,b_i^\ell,b_i^r\}$ must all
belong to one triconnected component, call it $D$.  Since $D$ is
3-connected, there must exist a path $P$ from $b_i^\ell$ to $b_i^r$
within $D-\{u,v\}$, and this is the path that we use for this
child-peripheral pair.

It remains to argue that these paths are disjoint.  
Let $\{b',b''\}$ be some other child-peripheral-pair, say at cutting pair
$\{u',v'\}$, such that $\{b',b'',u',v'\}$ belong to triconnected 
component $D'$ and we assigned a path $P'$ in $D'-\{u',v'\}$ to this 
child-peripheral pair.

Recall that cutting pair $\{u,v\}$ splits the graph into multiple
cut-components.  One of those is $C_i$, the child-component that
contained $b_i^\ell$ and $b_i^r$ and therefore also the triconnected
component $D$ and the path $P$.  We now distinguish cases
depending on which of the cut-components contains $D'$:
\begin{itemize}
\item $D'$ is part of a child-component of $\{u,v\}$ other than $C_i$.

	We know that child-cut-components   are vertex-disjoint
	except for $\{u,v\}$.  
	Therefore $D$ and $D'$ are vertex-disjoint except for perhaps $\{u,v\}$.
	Hence $P$ and $P'$ are vertex-disjoint.
\item $D'$ is part of the parent-component of $\{u,v\}$. 

	As before, since cut-components are vertex-disjoint except for
	$\{u,v\}$, this implies that $P$ and $P'$ are vertex-disjoint.

\item $D'$ is part of the child-component $C_j$ of $\{u,v\}$. 

	This implies that $\{u,v\}\neq \{u',v'\}$, since for each
	cutting pair, each cut-component gets only one peripheral pair.
	($u=u'$ or $v=v'$ is possible, but not both.)  Changing the
	point of view, now consider the cut-components of $\{u',v'\}$.
	Here $D'$ belongs to a child-component (because $\{b',b''\}$
	is a child-peripheral-pair), but $D$ belongs to the parent-component
	(since $D'$ belongs to a child-component of $\{u,v\}$).
	Exchanging the roles of the two cutting pairs hence shows 
	as in the
	previous case that $P$ and $P'$ are vertex-disjoint. \qed
\end{itemize}

\subsection{Making graphs connected} 

In this section, we give a proof of Lemma~\ref{lem:make_connected},
which states:

\begin{quotation}
Let $G$ be a planar graph. 
Then we can add edges to $G$ so that
the resulting graph $G'$ is planar, connected, and $pw(G')=\max\{1,pw(G)\}$.
\end{quotation}

Let $C_1,\dots,C_L$ be the connected components of $G$.  Each
of them has pathwidth at most $pw(G)$ since they are subgraphs of $G$;
let $\calP_i$ be a path decomposition of $C_i$ of width at most $pw(G)$.
Start with path decomposition $\calP_1$.  Append one new bag, into
which we insert one arbitrary vertex $v_1$ from the last bag of $\calP_1$
and one arbitrary vertex $u_2$ from the first bag of $\calP_2$.  Then 
append $\calP_2$.  Repeat with the remaining components: insert
a new bag after the last bag of $\calP_i$, give it one vertex $v_i$
from the last bag of $\calP_i$ and one vertex $u_{i+1}$ from
the first bag of $\calP_{i+1}$, and then append $\calP_{i+1}$.
Clearly we get a path decomposition $\calP$ of $G$ of width $\max\{1,pw(G)\}$.  

Define $G'$ to be the graph obtained by adding $(u_i,v_{i+1})$ to $G$,
for $i=1,\dots,L-1$.  Clearly $\calP$ is also a path decomposition
of $G'$, since we created bags for each of these new edges.  Also $G'$
is planar since adding an edge between two vertices in different
connected components cannot destroy planarity.   This shows the result.
\qed

\end{appendix}

\end{document}